\documentclass[aps,prb,reprint, amsmath, amssymb, superscriptaddress]{revtex4-1}
\usepackage[T1]{fontenc}
\usepackage[utf8]{inputenc}
\usepackage{CJK}
\usepackage{amsthm}
\usepackage{graphicx}
\usepackage{dcolumn}
\usepackage{bm}
\usepackage{mathptmx}
\usepackage{amsmath}

\makeatletter
\theoremstyle{plain}
\newtheorem{thm}{\protect\theoremname}

\usepackage{braket}

\renewcommand\[{\begin{equation}}
\renewcommand\]{\end{equation}}

\makeatother

\usepackage{babel}
\usepackage{xcolor}

\providecommand{\theoremname}{Theorem}

\providecommand{\tabularnewline}{\\}

\usepackage{hyperref}
\makeatother

\begin{document}
\global\long\def\u#1{\mathrm{\:#1}}

\title{Topological magnons for thermal Hall transport in frustrated magnets with bond-dependent interactions}

\author{Emily Z. Zhang}
\affiliation{Department of Physics, University of Toronto, Toronto, Ontario M5S 1A7, Canada}

\author{Li Ern Chern}
\affiliation{Department of Physics, University of Toronto, Toronto, Ontario M5S 1A7, Canada}

\author{Yong Baek Kim}
\affiliation{Department of Physics, University of Toronto, Toronto, Ontario M5S 1A7, Canada}


\begin{abstract}

Thermal transport in topologically-ordered phases of matter provides valuable insights as it can detect the charge-neutral quasiparticles that would not directly couple to electromagnetic probes. An important example is the edge heat transport of the Majorana fermions in a chiral spin liquid, which leads to a half-quantized thermal Hall conductivity. This signature is precisely what has recently been measured in $\alpha$-RuCl$_3$ under external magnetic fields. The plateau-like behaviour of the half-quantized thermal Hall conductivity as a function of external magnetic field, and the peculiar sign change depending on the magnetic field orientations, have been proposed to be strong evidence for the non-Abelian Kitaev spin liquid. On the other hand, for in-plane magnetic fields, it has been theoretically shown that such a sign structure can also arise from topological magnons in the field-polarized state. In this work, we investigate the full implications of topological magnons as heat carriers on thermal transport measurements. We first prove analytically that for any commensurate order with a finite magnetic unit cell, reversing the field direction leads to a sign change in the magnon thermal Hall conductivity in two-dimensional systems. We corroborate this proof numerically with nontrivial magnetic orders as well as the field-polarized state in Kitaev magnets subjected to an in-plane field. In the case of the tilted magnetic field, in which there exists both finite in-plane and out-of-plane field components, we find that the plateau-like behaviour of the thermal Hall conductivity and the sign change upon the reversal of the in-plane component of the magnetic field arises in the partially-polarized state, as long as the in-plane field contribution to the Zeeman energy is significant. While these results are consistent with the experimental observations, we comment on other aspects that require further investigation in future studies.

\end{abstract}


\maketitle

\section{Introduction}
The search for quantum spin liquid (QSL) states\cite{annurev-conmatphys-020911-125138,Savary_2016,RevModPhys.89.025003,Broholmeaay0668} with long range entanglement and emergent quasiparticle excitations is of both practical and fundamental interests. One avenue to achieve a QSL state is via magnetic frustrations from bond-dependent interactions, which has motivated the study of $4d/5d$ materials with strong spin-orbit coupling\cite{PhysRevLett.102.017205,PhysRevLett.105.027204,Katukuri_2014,PhysRevLett.112.077204,PhysRevB.90.041112,nphys3322,s41567-020-0874-0,annurev-conmatphys-031115-011319,Schaffer_2016,Winter_2017,s42254-019-0038-2,Janssen_2019}. These materials naturally possess bond-dependent interactions and have the potential to realize the $S = 1/2$ Kitaev model on a honeycomb lattice\cite{KITAEV20062}, and has garnered much intrigue due to it being exactly solvable and having a spin liquid ground state\cite{PhysRevLett.114.147201,nmat4604,Banerjee1055,PhysRevLett.119.037201,nphys4264,PhysRevLett.119.227202,s41535-018-0079-2}. A great leap forward was achieved in the search for the Kitaev spin liquid (KSL) when a half-quantized thermal Hall conductivity was measured in the material $\alpha$-RuCl$_3$ under external magnetic fields\cite{s41586-018-0274-0}. The half-quantization is a signature of Majorana fermions – the fractionalized excitations of the KSL\cite{KITAEV20062,PhysRevLett.119.127204,PhysRevX.8.031032,PhysRevLett.121.147201,PhysRevB.100.060405}. Their presence in a real material, once confirmed, would be a major breakthrough in spin liquid physics\cite{1809.08247,s41467-019-08459-9,PhysRevLett.123.197201,s41467-019-10405-8,PhysRevB.100.144445,PhysRevResearch.2.013014,s41467-020-15320-x,2007.07259,2009.03332}.

The half-quantized thermal Hall conductivity is observed in $\alpha$-RuCl$_3$ not only under tilted fields, but also under completely in-plane fields\cite{2001.01899,czajka2021oscillations}. This is unlike in ordinary metals, in which the electronic Hall effect only occurs when the field has a finite out-of-plane component. Such an anomalous Hall effect in $\alpha$-RuCl$_3$ is postulated to originate from Majorana fermions in the non-Abelian Kitaev spin liquid. For in-plane fields, a sign change in the thermal Hall conductivity was measured when the field was flipped from the $a$ to $-a$ direction, while no appreciable signals were measured when the field was applied along the $b$ direction (see Fig. \ref{fig:honeycomb}). Although this peculiar sign structure of the thermal Hall conductivity\cite{2012.11604} is consistent with the non-Abelian KSL scenario\cite{s41467-019-10405-8,2004.13723}, a recent theoretical study\cite{chern2020sign} demonstrated that it can also arise from topological magnons\cite{1.4959815,PhysRevB.94.094405,Owerre_2017,PhysRevB.98.060404,PhysRevB.98.060405} in the polarized state in Kitaev magnets. In the case of tilted fields in the $ac$ plane, experiments have measured a similar sign change in the thermal Hall conductivity when the tilting angle from the $c$-axis was switched from $-60$ degrees to $+60$ degrees in the suspected spin liquid regime\cite{2001.01899}. Determining whether these signals are uniquely caused by the non-Abelian QSL state is crucial, thus there is a pressing need to critically examine other possible mechanisms. 

\begin{figure}
\begin{centering}
\includegraphics[width=0.9\columnwidth]{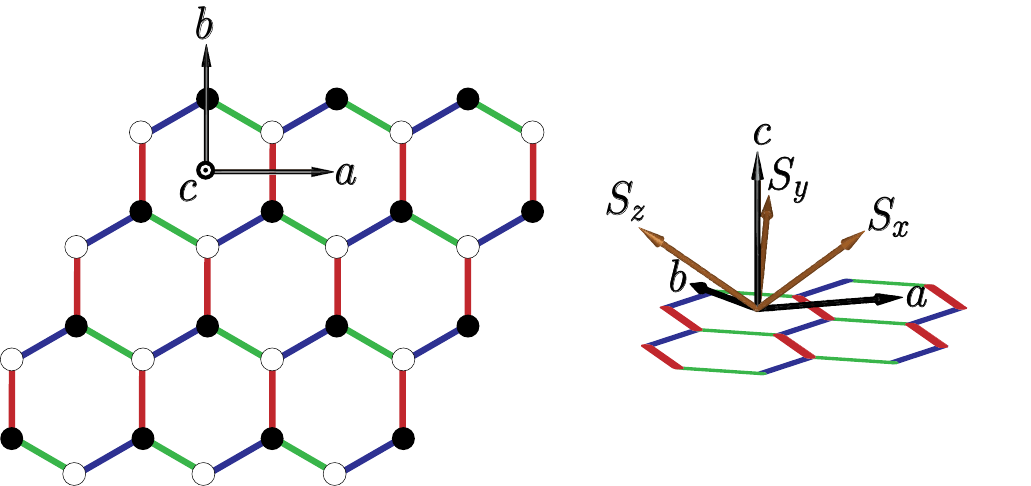}
\par\end{centering}
\caption{The $x$, $y$, and $z$ bonds in $K\Gamma\Gamma'$ model are coloured in blue, green, and red, respectively. The crystallographic $a$ (in-plane), $b$ (in-plane), and $c$ (out-of-plane) directions are indicated. In the cubic basis according to which the spin components in the $K\Gamma\Gamma'$ model are defined, the $a$, $b$, and $c$ directions are given by $[11\bar{2}]$, $[\bar{1}10]$, and $[111]$, respectively.
\label{fig:honeycomb}}
\end{figure}

In this work, we investigate the possibility of magnons as the heat carriers responsible for the thermal transport observed in experiments. First, we theoretically demonstrate that a sign change in the magnon thermal Hall conductivity follows from reversing the field direction. This generic property holds for any magnetic order with a finite unit cell, and for any field direction, in a bilinear spin model. A concise proof is presented in the main text. In a previous work, the magnon thermal Hall effect in the polarized state was explored for fields parallel to the $a$ and $b$ directions \cite{chern2020sign}. As an extension to this study, we consider the $K\Gamma\Gamma'$ model and determine the classical phase diagram in the presence of all possible in-plane magnetic field directions. We then show that the general sign change property indeed holds for various magnetic orders and polarized phases appearing in the phase diagram. 

Furthermore, we explicitly compute the thermal Hall conductivities for the $K\Gamma\Gamma'$ model under tilted fields. As the field increases, the system does not immediately polarize, but enters a ferromagnetic phase with a dominant in-plane magnetization\cite{PhysRevB.96.064430}. Therefore, right above the ferromagnetic phase transition, a sign change in the thermal Hall conductivity under the reversal of the in-plane field component can still be observed, as in the case of completely in-plane fields. For the two-dimensional zig-zag ordered state\cite{PhysRevB.91.144420,PhysRevB.92.235119} below the critical field, the sign change behaviour can also be observed. We show that certain aspects of magnon heat transport in the $K\Gamma\Gamma'$ model are consistent with experimental results, and highlight others that require further investigations. 

The rest of the paper is organized as follows. In Section II, we provide a proof for the sign structure of the magnon thermal Hall conductivity on a bilinear spin Hamiltonian. Section III presents a classical phase diagram with magnetic fields in the $ab$ plane as well as numerical results of the thermal Hall conductivities. Section IV explores the thermal Hall effect under tilted magnetic fields. We discuss the dependence of the overall sign of the thermal Hall conductivity on the model parameters in Section V. Lastly, Section VI discusses the key findings of this work and provides a future outlook. 

\section{Thermal Hall sign structure}
Consider a general bilinear spin Hamiltonian with a Zeeman field,
\begin{equation} \label{bilinearhamiltonian}
H=\sum_{ij}\mathbf{S}_{i}^{T}H_{ij}\mathbf{S}_{j}-\mathbf{h}^{T}\sum_{i}\mathbf{S}_{i}
\end{equation}
where $H_{ij}$ is a real symmetric matrix that encodes the interaction between the spins at sites $i$ and $j$.
Suppose that we have obtained the ground state spin configuration $\lbrace \mathbf{S}_i \rbrace$ of the classical model in \eqref{bilinearhamiltonian}.
To construct the linear spin wave Hamiltonian \cite{PhysRev.58.1098,Jones_1987} for this ground state, the spins are first passively rotated on each site of the lattice
such that the $z$-direction of the local coordinate system is parallel
to the spin orientation. We define the following rotation matrix 
\begin{equation} \label{rotationmatrix}
R_{i}=\left[\begin{matrix}\cos\theta_{i}\cos\phi_{i} & -\sin\phi_{i} & \sin\theta_{i}\cos\phi_{i}\\
\cos\theta_{i}\sin\phi_{i} & \cos\phi_{i} & \sin\theta_{i}\sin\phi_{i}\\
-\sin\theta_{i} & 0 & \cos\theta_{i}
\end{matrix}\right]
\end{equation}
where $\theta_{i}$ and $\phi_{i}$ are polar angles in the original
cubic $xyz$ coordinates, and 
$\left(S_{i}^{x},S_{i}^{y},S_{i}^{z}\right)=S\left(\sin\theta_{i}\cos\phi_{i},\sin\theta_{i}\sin\phi_{i},\cos\theta_{i}\right),$
such that
\[
\mathbf{S}_{i}=R_{i}\tilde{\mathbf{S}}_{i},
\]
where, classically, $\tilde{\mathbf{S}}_{i}=\left(0,0,S\right)$ is the spin at $i$ measured in the local coordinate system.

The rotated, site-dependent Hamiltonian is then given by 

\[
H=\sum_{ij}\mathbf{\tilde{S}}_{i}^{T}\tilde{H}_{ij}\mathbf{\tilde{S}}_{j}-\sum_{i}\mathbf{\tilde{h}}_{i}^{T}\mathbf{\tilde{S}}_{i}
\]
where $\tilde{H}_{ij}=R_{i}^{T}H_{ij}R_{j}$ and $\tilde{\mathbf{h}}_{i}=R_{i}^T\mathbf{h}$. Note that all the elements of $\tilde{H}_{ij}$ are real. 

In the rotated basis, we perform a Holstein-Primakoff expansion \cite{PhysRev.58.1098},
\begin{align*}
\tilde{S}_{i}^{z} & =S-b_{i}^{\dag}b_{i}=S-n_{i}\\
\tilde{S}_{i}^{x} & =\frac{\sqrt{2S-n_{i}}b_{i}+b_{i}^{\dag}\sqrt{2S-n_{i}}}{2}\approx\sqrt{\frac{S}{2}}\left(b_{i}+b_{i}^{\dag}\right)\\
\tilde{S}_{i}^{y} & =\frac{\sqrt{2S-n_{i}}b_{i}-b_{i}^{\dag}\sqrt{2S-n_{i}}}{2}\approx-i\sqrt{\frac{S}{2}}\left(b_{i}-b_{i}^{\dag}\right).
\end{align*}
Keeping only terms that contribute up to quadratic order in the Hamiltonian, we get
\begin{align*}
H & =\sum_{\left\langle i,j\right\rangle }\frac{S}{2}\left[\begin{matrix}b_{i} & b_{i}^{\dag}\end{matrix}\right]\left[\begin{matrix}1 & -i\\
1 & i
\end{matrix}\right]\left[\begin{matrix}\tilde{H}_{ij}^{11} & \tilde{H}_{ij}^{12}\\
\tilde{H}_{ij}^{21} & \tilde{H}_{ij}^{22}
\end{matrix}\right]\left[\begin{matrix}1 & 1\\
-i & i
\end{matrix}\right]\left[\begin{matrix}b_{j}\\
b_{j}^{\dag}
\end{matrix}\right]\\
 &-\sum_{\left\langle i,j\right\rangle }S\tilde{H}_{ij}^{33}\left(b_{i}^{\dag}b_{i}+b_{j}^{\dag}b_{j}\right)+\sum_{i}\left(\tilde{\mathbf{h}}_{i}\right)_{3}b_{i}^{\dag}b_{i}.
\end{align*}
Expanding out the first term, we obtain
\footnotesize
\begin{align}
	H & = \sum_{\left\langle i,j\right\rangle }\frac{S}{2}
	\left[\begin{matrix}b_{i} & b_{i}^{\dag}\end{matrix}\right]\times \nonumber\\
	&\left[\begin{matrix}\tilde{H}_{ij}^{11}-\tilde{H}_{ij}^{22}-i\left(\tilde{H}_{ij}^{12}+\tilde{H}_{ij}^{21}\right)
	& \tilde{H}_{ij}^{11}+\tilde{H}_{ij}^{22}+i\left(\tilde{H}_{ij}^{12}-\tilde{H}_{ij}^{21}\right)\\
	\tilde{H}_{ij}^{11}+\tilde{H}_{ij}^{22} -i\left(\tilde{H}_{ij}^{12}-\tilde{H}_{ij}^{21}\right)
	& \tilde{H}_{ij}^{11}-\tilde{H}_{ij}^{22}+i\left(\tilde{H}_{ij}^{12}+\tilde{H}_{ij}^{21}\right)
	\end{matrix}\right]
	\left[\begin{matrix}b_{j} \nonumber \\
	b_{j}^{\dag}
	\end{matrix}\right]\\
 	& -\sum_{\left\langle i,j\right\rangle }S\tilde{H}_{ij}^{33}\left(b_{i}^{\dag}b_{i}+b_{j}^{\dag}b_{j}\right)+\sum_{i}	\left(\tilde{\mathbf{h}}_{i}\right)_{3}b_{i}^{\dag}b_{i}.
\end{align}
\normalsize

We assume that the ground state spin configuration is a magnetic order with a finite unit cell of $\mathcal{N}$ sublattices, so that it has a real space periodicity that allows us to perform the Fourier transform, which yields
\footnotesize
\begin{align}
H & =\frac{S}{4}\sum_{\underset{t\in X_t}{s\in X_s}} \sum_{\mathbf{k}}\left(\left(\tilde{H}_{st}^{11}-\tilde{H}_{st}^{22}-i\left(\tilde{H}_{st}^{12}+\tilde{H}_{st}^{21}\right)\right)e^{i\mathbf{k}\cdot\left(\mathbf{R}_{s}-\mathbf{R}_{t}\right)}b_{\mathbf{k},s}b_{\mathbf{-k},t}+\text{c.c.}\right)\nonumber \\
 & +\frac{S}{4}\sum_{\underset{t\in X_t}{s\in X_s}}\sum_{\mathbf{k}}\left(\left(\tilde{H}_{st}^{11}+\tilde{H}_{st}^{22}+i\left(\tilde{H}_{ij}^{12}-\tilde{H}_{ij}^{21}\right)\right)e^{i\mathbf{k}\cdot\left(\mathbf{R}_{s}-\mathbf{R}_{t}\right)}b_{\mathbf{k},s}b_{\mathbf{k},t}^{\dag}+\text{c.c.}\right)\nonumber \\
 & -\frac{S}{2}\sum_{\underset{t\in X_t}{s\in X_s}}\sum_{\mathbf{k}}\tilde{H}_{st}^{33}\left(b_{\mathbf{k},s}^{\dag}b_{\mathbf{k},s}+b_{\mathbf{k},t}^{\dag}b_{\mathbf{k},t}\right)+\sum_{s\in X_s,\mathbf{k}}\left(\tilde{\mathbf{h}}_{s}\right)_{3}b_{\mathbf{k},s}^{\dag}b_{\mathbf{k},s}.\label{eq:ham}
\end{align}
\normalsize
where $s$ and $t$ label the sublattices of an interacting pair of spins, which belong to the magnetic unit cells $X_s$ and $X_t$ respectively. Factors of 1/2 are inserted to avoid double counting. Then, we can write the spin wave Hamiltonian in the form of 
\[
H/S=\sum_{\mathbf{k}}\Psi_{\mathbf{k}}^{\dag}\mathbf{D}_{\mathbf{k}}\Psi_{\mathbf{k}},
\]
where the spinor $\Psi_{\mathbf{k}}=\left(b_{\mathbf{k},1},\dots,b_{\mathbf{k},\mathcal{N}},b_{-\mathbf{k},1}^{\dag},\dots,b_{-\mathbf{k},\mathcal{N}}^{\dag}\right)$
, and $\mathbf{D}_{\mathbf{k}}$ is a $2\mathcal{N}\times2\mathcal{N}$
matrix of the form 
\[
\mathbf{D}_{\mathbf{k}}=\left[\begin{matrix}\mathbf{A_{k}} & \mathbf{B_{k}}\\
\mathbf{B_{-k}^{*}} & \mathbf{A_{-k}^{T}}
\end{matrix}\right]
\label{eq:dk}
\]
with $\mathbf{A_{k}}$ and $\mathbf{B_{k}}$ being $\mathcal{N}$ dimensional
matrices.
Using this linear spin wave formalism, we present the following theorem:
\begin{thm}
When the direction of the external magnetic field $\mathbf{h}$ in \eqref{bilinearhamiltonian} is reversed, there is a sign change in the magnon thermal Hall conductivity $\kappa_{xy}$ for any magnetic order with a finite unit cell.
\end{thm}

\begin{proof}
For the Hamiltonian \eqref{bilinearhamiltonian}, when the Zeeman field transforms as $\mathbf{h}\mapsto-\mathbf{h}$,
the spins in the ground state undergo a transformation of
\begin{align*}
\mathbf{S}_{i} & \mapsto-\mathbf{S}_{i}
\end{align*}
or equivalently
$
\left\{ \theta_{i},\phi_{i}\right\} \mapsto\left\{ \pi-\theta_{i},\phi_{i}+\pi\right\}
$, such that the ground state energy remains unchanged. In other words, if $\lbrace \mathbf{S}_i \rbrace$ is the ground state spin configuration under the field $\mathbf{h}$, then $\lbrace -\mathbf{S}_i \rbrace$ is the ground state spin configuration under the field $-\mathbf{h}$.

Therefore, when $\mathbf{h}\mapsto-\mathbf{h}$, the rotation matrix \eqref{rotationmatrix} undergoes a transformation from
\[
R_{i}^{h}=\left[\begin{matrix}\cos\theta_{i}\cos\phi_{i} & -\sin\phi_{i} & \sin\theta_{i}\cos\phi_{i}\\
\cos\theta_{i}\sin\phi_{i} & \cos\phi_{i} & \sin\theta_{i}\sin\phi_{i}\\
-\sin\theta_{i} & 0 & \cos\theta_{i}
\end{matrix}\right]
\]
to 
\begin{align}
R_{i}^{\bar{h}}  =\left[\begin{matrix}\cos\theta_{i}\cos\phi_{i} & \sin\phi_{i} & -\sin\theta_{i}\cos\phi_{i}\\
\cos\theta_{i}\sin\phi_{i} & -\cos\phi_{i} & -\sin\theta_{i}\sin\phi_{i}\\
-\sin\theta_{i} & 0 & -\cos\theta_{i}
\end{matrix}\right].
\end{align}

The rotated spin Hamiltonians under the fields $\mathbf{h}$ and $\mathbf{-h}$ are given by
\begin{align*}
\tilde{H}_{ij}^{h} & =\left(R_{i}^{h}\right)^{T}H_{ij}R_{j}^{h} = \left[\begin{matrix}\tilde{H}_{ij}^{11} & \tilde{H}_{ij}^{12} & \tilde{H}_{ij}^{13} \\ \tilde{H}_{ij}^{21} & \tilde{H}_{ij}^{22} & \tilde{H}_{ij}^{23}\\
\tilde{H}_{ij}^{31} & \tilde{H}_{ij}^{32} & \tilde{H}_{ij}^{33}
\end{matrix}\right] , \\
\tilde{H}_{ij}^{\bar{h}} & =\left(R_{i}^{\bar{h}}\right)^{T}H_{ij}R_{j}^{\bar{h}} =\left[\begin{matrix}\tilde{H}_{ij}^{11} & -\tilde{H}_{ij}^{12} & -\tilde{H}_{ij}^{13}\\
-\tilde{H}_{ij}^{21} & \tilde{H}_{ij}^{22} & \tilde{H}_{ij}^{23}\\
\tilde{-H}_{ij}^{31} & \tilde{H}_{ij}^{32} & \tilde{H}_{ij}^{33}
\end{matrix}\right] .
\end{align*}
Therefore, when the field direction is reversed, the
only components of the rotated Hamiltonian relevant to linear spin
wave theory that pick up a sign change are $\tilde{H}_{ij}^{12}$
and $\tilde{H}_{ij}^{21}$. For the rotated field, the Zeeman term contribution to the linear spin wave theory remains unchanged. From equations (\ref{eq:ham}) and (\ref{eq:dk}), we find that $\mathbf{D}_{\mathbf{k}}^{\bar{h}}=\left(\mathbf{D}_{\mathbf{-k}}^{h}\right)^{*}$. If the system is two-dimensional, we can show that $\kappa_{xy}$ changes sign, and the details of the proof can be found in a previous work\cite{chern2020sign}.
\end{proof}

Theorem 1 holds for any magnetic order with a finite unit cell on the underlying Bravais lattice, and for any field direction. It is independent of the details of the spin interactions $H_{ij}$ as long as it is bilinear. In particular, the theorem holds for any order of nearest neighbour interaction, including the first- and third-nearest neighbour Heisenberg interactions, $J$ and $J_3$; we only require the ground state to exhibit a translational symmetry (defined by the magnetic unit cell), such that we are able to define the linear spin wave Hamiltonian in the momentum space \eqref{eq:ham}. A situation in which this theorem breaks down is a magnetic order that is incommensurate with the lattice and thus devoid of translational symmetry. In the following sections, we will be exploring this general sign change property in the $K\Gamma\Gamma'$ model.

\section{Classical phase diagram with the magnetic field in the $ab$-plane and thermal Hall conductivity}

We consider the $K\Gamma\Gamma'$ model under a magnetic field as
a model for $\alpha$-RuCl$_{3}$, which is given by $H=\sum_{\left\langle ij\right\rangle \in\lambda}\mathbf{S}_{i}^{T}H_{\lambda}\mathbf{S}_{j}-\mathbf{h}^{T}\sum_{i}\mathbf{S}_{i}$
where 
\[
H_{x}=\left[\begin{matrix}K & \Gamma' & \Gamma'\\
\Gamma' & 0 & \Gamma\\
\Gamma' & \Gamma & 0
\end{matrix}\right],H_{y}=\left[\begin{matrix}0 & \Gamma' & \Gamma\\
\Gamma' & K & \Gamma'\\
\Gamma & \Gamma' & 0
\end{matrix}\right],H_{z}=\left[\begin{matrix}0 & \Gamma & \Gamma'\\
\Gamma & 0 & \Gamma'\\
\Gamma' & \Gamma' & K
\end{matrix}\right].
\]
First, we chose a model with $K=-1$, $\Gamma=0.5$, and $\Gamma'=-0.02$, and we numerically
explored the classical ground states using simulated annealing\cite{PhysRevLett.117.277202,PhysRevResearch.2.013014} with in-plane magnetic fields $\mathbf{h}=h\left(\cos\theta\left[11\bar{2}\right]+\sin\theta\left[\bar{1}10\right]\right)$,
with $h\in\left[0,0.3\right]$ and $\theta\in\left[0,2\pi\right)$.
These computations resulted in a rich phase diagram shown in Fig.
\ref{fig:Classical-phase-diagram}a, which notably includes non-trivial
intermediate phases between zig-zag orders and the polarized state
-- the largest of which contains 10 sites per magnetic unit cell. Neutron
scattering experiments have been performed to obtain the phase diagram
for a field applied parallel to the $a$ axis, but not the $b$ axis\cite{balz2020fieldinduced}. Thus, a potential application of this in-plane phase diagram is to serve as an initial predictor for future experiments. 

\begin{figure}
\begin{centering}
\includegraphics[width=1\columnwidth]{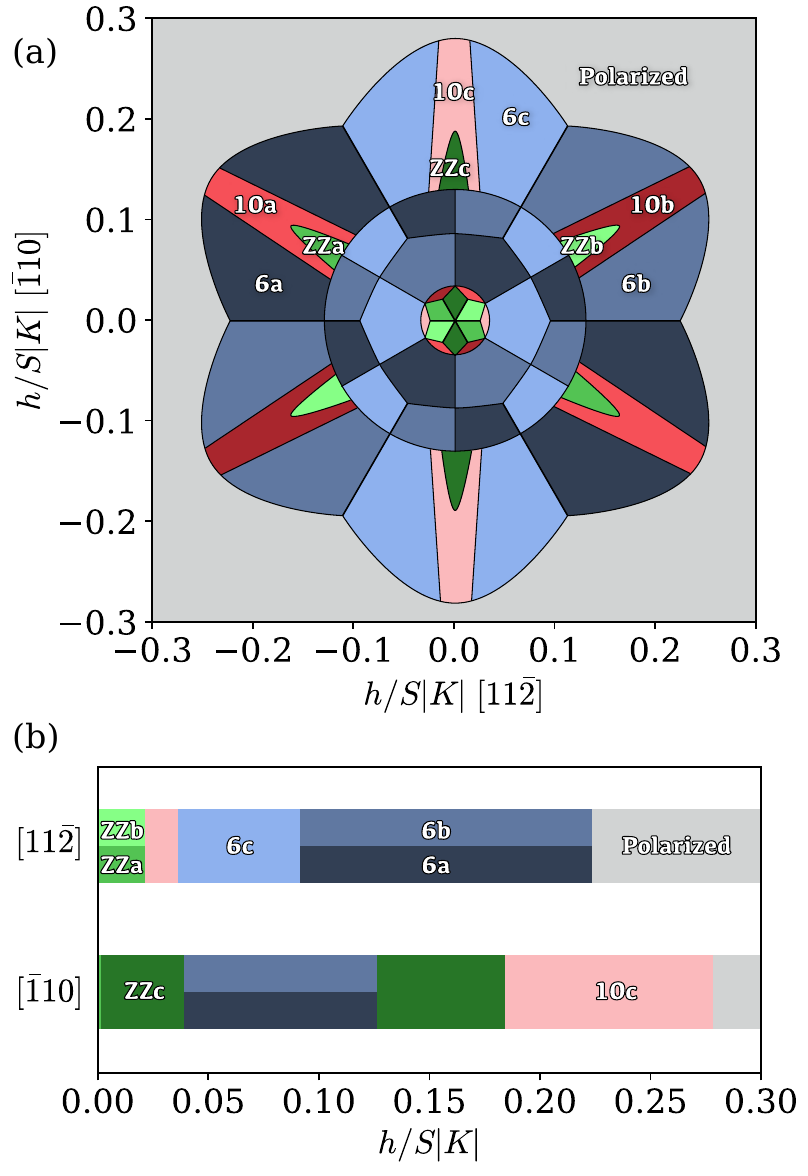}
\par\end{centering}
\caption{(a) Classical phase diagram of the $K\Gamma\Gamma'$ model under a magnetic field $\mathbf{h}$ in the honeycomb plane. The interaction parameters are fixed to be $(K,\Gamma,\Gamma') =(-1,0.5,-0.02)$. The polarized state is indicated by grey, the zig-zag (ZZ) order by shades of green, the 6-site order by shades of blue, and the 10-site order by shades of red. The a, b, and c configurations (not to be confused with the $a$, $b$, and $c$ crystallographic directions) of each order are shown with different shadings of the same color. (b) Classical phase diagrams under magnetic fields along the $[11\bar{2}]$ (upper) and $[\bar{1}10]$ (lower) directions, which are extracted from (a). For the $[11\bar{2}]$ field, ZZa and ZZb orders are degenerate. The 6a and 6b orders are degenerate in both cases. \label{fig:Classical-phase-diagram}}
\end{figure}

There are three distinct configurations for each of the ZZ, 6-site\cite{1807.06192}, and 10-site orders, examples of which are shown in the Supplemental Materials\cite{SM}. Note that when the field is along the $c$-direction\cite{PhysRevResearch.2.013014}, the three configurations of each order are degenerate and related by the $C_3$ symmetry of the system. For in-plane fields, the $C_{3}$ symmetry is broken, such that the three configurations of each order in general differ in energy. However, when the in-plane field is applied along high-symmetry directions that are equivalent to the $a$ or $b$ direction, two configurations of the same order may be degenerate. For example, 6a and 6b are degenerate under magnetic fields along the $a$ and $b$ directions, see Fig.~\ref{fig:Classical-phase-diagram}b.

\begin{figure}
\begin{centering}
\includegraphics[width=0.8\columnwidth]{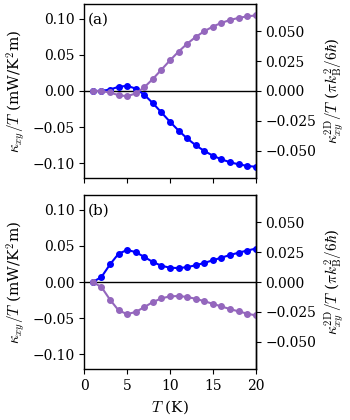}
\par\end{centering}
\caption{Thermal Hall conductivity $\kappa_{xy}/T$ as a function
of temperature $T$ due to magnons in (a) the 6-site order at $h/S|K|=0.2$ and (b) the polarized state at $h/S|K|=0.22$, with the interaction parameters $\left(K,\Gamma,\Gamma'\right)=\left(-1,0.5,-0.02\right)$. The magnetic field is applied along the $a$ ($-a$) direction for the blue (purple) curves, and the field strengths are chosen to border the critical field. The two-dimensional thermal Hall conductance $\kappa_{xy}^{2D}/T\equiv\kappa_{xy}d/T$ is also indicated. \label{fig:inplane-temp}}
\end{figure}

\begin{figure}
\begin{centering}
\includegraphics[width=0.9\columnwidth]{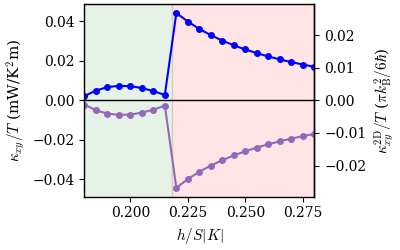}
\par\end{centering}
\caption{Thermal Hall conductivity $\kappa_{xy}/T$ due to magnons as a function
of field strength $h/S|K|$ at temperature $T=5 \u{K}$, with the interaction parameters $\left(K,\Gamma,\Gamma'\right)=\left(-1,0.5,-0.02\right)$. The magnetic field is applied along the $a$ ($-a$) direction for the blue (purple) curves. There is a phase transition from the 6-site order (green region) to the polarized state (red region) as the field increases. The two dimensional thermal Hall conductance $\kappa_{xy}^{2D}/T\equiv\kappa_{xy}d/T$ is also indicated. \label{fig:inplane-field}}
\end{figure}
To calculate $\kappa_{xy}$ due to magnons, we first obtained the linear spin wave dispersion $\varepsilon_{n\mathbf{k}}$ by diagonalizing (\ref{eq:dk}) using a Bogoliubov transformation\cite{Bogoliubov1947}. We then computed the Berry curvature $\boldsymbol{\Omega}_{n\mathbf{k}}$ using the Bogoliubov transformation matrices $\mathbf{T}_\mathbf{k}$. $\kappa_{xy}$ is given by\cite{PhysRevLett.106.197202,PhysRevB.89.054420,JPSJ.86.011010}
\begin{align}
	\kappa_{xy} = -\frac{k_B^2 T}{\hbar V}
	\sum_n  \sum_{\mathbf{k} \in \text{FBZ} }
	\left\{ c_2 \left[ g\left(\varepsilon_{n\mathbf{k}} \right) \right] - 			\frac{\pi^2}{3} \right\} 	\boldsymbol{\Omega}_{n\mathbf{k}},
\end{align}
where FBZ is the crystal first Brillouin zone, $c_2(x)=(1+x)\{\ln [(1+x)/x]\}^2 - (\ln (x))^2 - 2\text{Li}_2(-x)$, and $g$ is the Bose-Einstein distribution. To ensure the convergence of $\kappa_{xy}$ at system size $L$, we also checked that the Chern number
\begin{align}
	C_n = \frac{1}{2\pi} \sum_{\mathbf{k}} \frac{(2\pi)^2}{A} \boldsymbol{\Omega}_{n\mathbf{k}},
\end{align}
where $A$ is the total area of the system, converged to an integer for each magnon band. Topological magnons are indicated by finite Chern numbers. For example, when the field was applied along the $a$ direction, we found that the two magnon bands of the polarized state carried the Chern numbers $\pm 1$, while the six magnon bands of the 6-site order carried the Chern numbers $0,0,0,1,-1,0$. 


We set the spin magnitude to be $S=1/2$ in the linear spin wave theory and calculated the magnon thermal Hall conductivity. We assumed the strength of the Kitaev interaction to be $\left|K\right|=80\u K$\cite{PhysRevB.93.155143,PhysRevLett.120.217205}, and set the interlayer distance to be $d=5.72\ \text{\AA}$ for $\alpha$-RuCl$_{3}$\cite{PhysRevB.92.235119,PhysRevLett.120.217205,s41586-018-0274-0,2001.01899}. We present the thermal Hall conductivities for the 6-site order and
the polarized state as a function of temperature in Figs. \ref{fig:inplane-temp}(a)-(b),
and as a function of field strength across the phase boundary in Fig. \ref{fig:inplane-field}, when the field was applied along the $a$ direction. The field strengths $h$ in each plot were chosen to be close to
the critical field separating the two magnetic orders. We make some observations from these results. First, large unit cell orders like the 6-site order can give rise to a finite thermal Hall effect. Additionally, the thermal Hall conductivity of the polarized state is much larger than that of the 6-site order for the parameterization we chose; a similar pattern can be seen in the experimental results. Furthermore, the thermal Hall conductivities of both orders under magnetic fields along the $a$ and $-a$ direction are equal in magnitude and opposite in sign, as predicted by Theorem 1. 

\begin{figure}
\begin{centering}
\includegraphics[width=0.8\columnwidth]{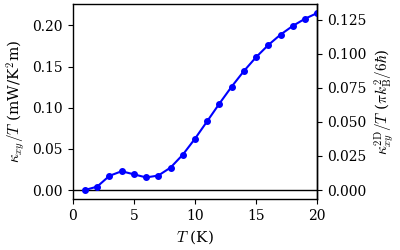}
\par\end{centering}
\caption{Thermal Hall conductivity $\kappa_{xy}/T$ as a function
of temperature $T$ due to magnons in the 6-site order under a magnetic field along the $b$ direction, with the interaction parameters $\left(K,\Gamma,\Gamma'\right)=\left(-1,0.5,-0.02\right)$ and the field strength $h/S|K|=0.1$. The two dimensional thermal Hall conductance $\kappa_{xy}^{2D}/T\equiv\kappa_{xy}d/T$ is also indicated. \label{fig:6b-temp}}
\end{figure}

We remark that the thermal Hall conductivity is not always zero when the external field is applied along the $b$ direction, as is for the polarized state. The lack of a signal for the polarized state is a special case that arises from a $C_2$ rotational symmetry about the $b$ axis in both the Hamiltonian and the ground state spin configuration. Such a $C_2$ symmetry can be broken by other magnetic orders, for example the 6-site order, which shows a finite thermal Hall conductivity when the field is applied along the $b$ direction (see Fig. \ref{fig:6b-temp}).

\section{Thermal Hall conductivity in the presence of tilted fields}

As explored in Sec.~II, there is a well-established
relation between the reversal of the field direction and the sign change of the magnon $\kappa_{xy}$. However, for a tilted field, the effect of reversing only the in-plane component $h_\parallel$, while fixing the out-of-plane component $h_\perp$, on the magnon $\kappa_{xy}$ is unclear. This is because the ground state spin configurations of $(h_\parallel,h_\perp)$ and $(-h_\parallel,h_\perp)$ have no direct relation in general, unlike those of $\mathbf{h}$ and $-\mathbf{h}$ which are related by flipping of the spins.

In the experiment\cite{2001.01899}, the tilted field was applied in the $ac$-plane. The tilting angle $\theta$ is defined to be the angle between the field and the $c$-axis. A sign change in $\kappa_{xy}$ was experimentally observed in the proposed spin liquid regime when $\theta$ was changed from $60^{\circ}$ and $-60^{\circ}$. In this section, we propose a scenario where such a sign change can happen in the magnon $\kappa_{xy}$. The underlying magnetically ordered ground state is the partially-polarized ferromagnet (PPF), which will be described below.

For a tilted field, assuming classical spins, the system enters the ferromagnetic state at sufficiently high fields, but is never completely polarized at finite fields. This effect originates from the competition between the $\Gamma$ interaction and the external field\cite{PhysRevB.96.064430}. We refer to such a ferromagnetic state as the PPF. As in the experiment, we applied fields in the $ac$-plane with tilting angles $60^{\circ}$ and $-60^{\circ}$ measured from the $c$-axis. Using the parametrization $\left(K,\Gamma,\Gamma'\right)=(-1,0.2,-0.02)$, we find that right after the system enters the PPF, the magnetization is largely in-plane, as shown in Tables \ref{tab:isotropic} and \ref{tab:anisotropic}. As the field strength increases, the in-plane (out-of-plane) component of the magnetization gradually decreases (increases).

\begin{table}
\begin{centering}
\begin{tabular}{cccc}
$h/S|K|$ & $S_{a}/S$ & $S_{b}/S$ & $S_{c}/S$\tabularnewline
\hline 
0.10 & 0.996103 & 0 & 0.088193\tabularnewline
0.12 & 0.994717 & 0 & 0.102656\tabularnewline
0.14 & 0.993218 & 0 & 0.116265\tabularnewline
0.16 & 0.991633 & 0 & 0.129088\tabularnewline
0.18 & 0.989983 & 0 & 0.141185\tabularnewline
\end{tabular}
\par\end{centering}
\caption{Magnetization of the partially-polarized phase as a function
of magnetic field. The spin $\mathbf{S}=(S_{a},S_{b},S_{c})$ is given in the crytallographic basis. We use the interaction parameters $(K,\Gamma,\Gamma')=(-1,0.2,-0.02)$ and the isotropic $g$ tensor $(g_\parallel,g_\perp)=(1,1)$. The field is tilted by $\theta=60^{\circ}$ from the $c$-axis towards the $a$-axis. The field range is chosen to be near the critical field of ferromagnetic transition.\label{tab:isotropic}}
\end{table}

\begin{table}
\begin{centering}
\begin{tabular}{cccc}
$h/S|K|$ & $S_{a}/S$ & $S_{b}/S$ & $S_{c}/S$\tabularnewline
\hline 
0.05 & 0.998427 & 0 & 0.056059\tabularnewline
0.07 & 0.997302 & 0 & 0.073410\tabularnewline
0.09 & 0.996064 & 0 & 0.088640\tabularnewline
0.11 & 0.994774 & 0 & 0.102106\tabularnewline
0.13 & 0.993470 & 0 & 0.114090\tabularnewline
\end{tabular}
\par\end{centering}
\caption{Magnetization of the partially-polarized phase as a function
of magnetic field. The spin $\mathbf{S}=(S_{a},S_{b},S_{c})$ is given in the crytallographic basis. We use the interaction parameters $(K,\Gamma,\Gamma')=(-1,0.2,-0.02)$ and the anisotropic $g$ tensor $(g_\parallel,g_\perp)=(2.3,1.3)$. The field is tilted by $\theta=60^{\circ}$ from the $c$-axis towards the $a$-axis. The field range is chosen to be near the critical field of ferromagnetic transition. \label{tab:anisotropic}}
\end{table}

\begin{figure*}
\includegraphics[width=0.7\textwidth]{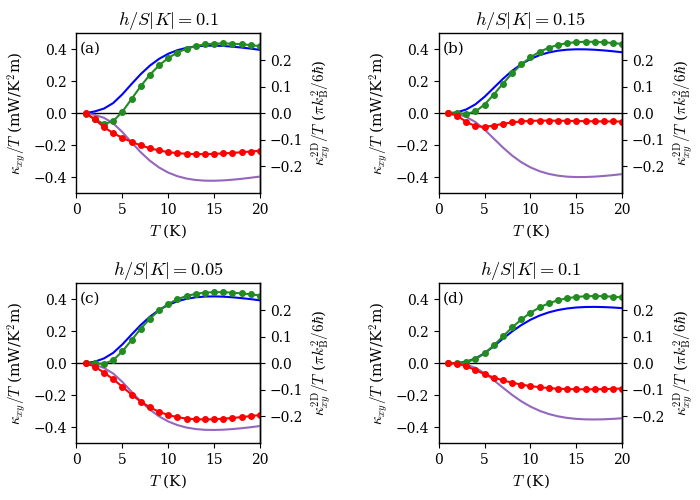}
\caption{Thermal Hall conductivity $\kappa_{xy}/T$ as a function of temperature $T$ due to magnons in various ferromagnetic states, with the interaction parameters $\left(K,\Gamma,\Gamma'\right)=(-1,0.2,-0.02)$. The magnetic field is applied in the $ac$-plane, $\mathbf{h}=hg_{\perp}\cos\theta\left[111\right]+hg_{\parallel}\sin\theta\left[11\bar{2}\right]$. An isotropic $g$ tensor $(g_\parallel, g_\perp)=(1,1)$ is used in (a) and (b), while an anisotropic $g$ tensor $(g_{\parallel}$, $g_{\perp})=(2.3,1.3)$ is used in (c) and (d). In each plot, the green (red) curve represents the data of the partially polarized ferromagnet with an almost in-plane magnetiztion, under a magnetic field with the tilting angle $\theta=60^{\circ}$ ($-60^{\circ}$). On the other hand, the blue (purple) curve represents the data of the ferromagnet with a completely in-plane magnetization (which is not the ground state, but plotted for comparison), under a magnetic field with the tilting angle $\theta=60^{\circ}$ ($-60^{\circ}$). The two dimensional thermal Hall conductance $\kappa_{xy}^{2D}/T\equiv\kappa_{xy}d/T$ is also indicated. \label{fig:tilted-temp}}
\end{figure*}

We plot the thermal Hall conductivities
as a function of temperature for isotropic (Figs.~\ref{fig:tilted-temp}a and \ref{fig:tilted-temp}b)
and anisotropic (Figs.~\ref{fig:tilted-temp}c and \ref{fig:tilted-temp}d) $g$-tensors. The $g$-tensor effectively modifies the field as $\mathbf{h}=hg_{\perp}\cos\theta\left[111\right]+hg_{\parallel}\sin\theta\left[11\bar{2}\right]$, where we chose $(g_{\parallel},g_{\perp})=(1,1)$ and $(g_{\parallel},g_{\perp})=(2.3,1.3)$ for the isotropic and anisotropic cases respectively\cite{Yadav2016,Winter2018,PhysRevB.94.064435}. The fields in \ref{fig:tilted-temp}a and \ref{fig:tilted-temp}c were chosen to be near the respective critical fields, such that the magnetizations of the PPFs were almost in-plane, i.e.~$S_a/S \approx \pm 1, S_c/S \approx 0$. We then expect the thermal Hall signals of these PPFs (e.g.~curves with dots in Fig.~\ref{fig:tilted-temp}c) to be similar to those of the ferromagnetic states with completely in-plane magnetizations, i.e.~$S_a/S = \pm 1, S_c/S = 0$ (e.g.~curves without dots in Fig.~\ref{fig:tilted-temp}c).
As we increase the field, the out-of-plane magnetization of the PPF becomes more significant, resulting in a larger deviation in $\kappa_{xy}$ from the case of a completely in-plane magnetization. Such a deviation is more prevalent in the case of isotropic $g$-tensor. This result is expected, since a heavier weight is placed on the in-plane component $h_\parallel$ of the field in the anisotropic case, which would result in a slower polarization as seen in Table \ref{tab:anisotropic}. Therefore, it is not guaranteed that the magnon thermal Hall conductivity will have exactly the same magnitude and opposite sign by changing $\theta$ from $60^\circ$ to $-60^\circ$. However, Fig.~\ref{fig:tilted-temp}c shows a regime in which we can obtain an approximate version of this behaviour.

\begin{figure}
\begin{centering}
\includegraphics[width=0.9\columnwidth]{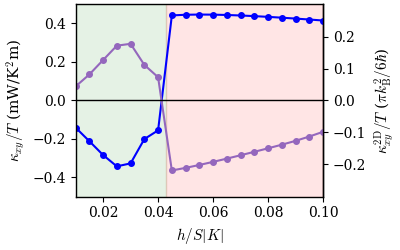}
\par\end{centering}
\caption{Thermal Hall conductivity $\kappa_{xy}/T$ due to magnons as a function of field strength $h/S|K|$ at temperature $T=15\protect\u K$, with the interaction parameters $(K,\Gamma,\Gamma')=(-1,0.2,-0.02)$. The magnetic field is applied along the $a$ (-$a$) direction for the blue (purple) curve. We use an anisotropic $g$ tensor ($g_{\parallel}=2.3$, $g_{\perp}=1.3$). There is a phase transition from the zig-zag order (green region) to the partially polarized state (red region) as the field increases. The two dimensional thermal Hall conductance $\kappa_{xy}^{2D}/T\equiv\kappa_{xy}d/T$ is also indicated. \label{fig:tilted-field}}
\end{figure}

In Fig. \ref{fig:tilted-field}, we plot the thermal Hall conductivities
for a zig-zag ordered state and a PPF state near the critical field. Some of these results are consistent with the thermal Hall signals that have been experimentally measured\cite{2001.01899}. We observe that the maximum PPF signal is larger than that of the zig-zag. It can be shown that this behaviour can be further optimized by tuning the system parameters such that the zig-zag signal is oppressed, while the PPF signal is preserved. Additionally, the PPF signal decays more rapidly as the field strength increases for the $+60^\circ$ field than the $-60^\circ$ field, which is also observed. Notably, the $+60^\circ$ signal as well as the $-60^\circ$ field also exhibits a plateau-like behaviour in the PPF region (the proposed spin liquid region in Ref.~\onlinecite{2001.01899}). However, there are also a few aspects of our result that still require further investigation. In particular, we observe a sign change in the zig-zag region near the critical field, which is expected due to the same argument applied for the PPF state, whereas no sign change was observed in the experiment. A possible explanation for this discrepancy is that we did not consider any 3-dimensional inter-layer ordering in our model. A recent neutron scattering experiment\cite{balz2020fieldinduced} revealed that there is a stacking of the layers in the zig-zag order, therefore if magnons are the predominant heat carrier, then the lack of a sign change most likely arises from this 3D order, which remains to be explored in a future work.

\section{Overall sign of $\kappa_{xy}$}

We have proven in Sec.~II that the sign change of $\kappa_{xy}$ upon reversing the field direction is universal (Theorem 1). However, the overall sign of $\kappa_{xy}$ for a given field direction depends on the choice of model parameters. For example, if we choose $(K,\Gamma,\Gamma')=(-1,0.5,-0.1)$ -- a different yet realistic parameterization of $\alpha$-RuCl$_3$ -- the overall sign of $\kappa_{xy}$ in the polarized state is opposite to those obtained from our earlier calculations. As shown in Fig. \ref{fig:thc-signs}, $\kappa_{xy}$ is negative (positive) when the field is applied along the $a$ ($-a$) direction, which is consistent with the overall sign observed in Ref.~\onlinecite{2001.01899}. This is unlike the sign of $\kappa_{xy}$ due to Majorana fermions in the non-Abelian Kitaev spin liquid, which is solely determined by the field direction\cite{KITAEV20062}. As we tune the interaction parameters, apart from a possible change of the overall sign of $\kappa_{xy}$, all other phenomenology remains the same. We have demonstrated that while the relative sign of $\kappa_{xy}$ is universal, its overall sign is parameter-dependent, and our analyses thus far are applicable to the experiments.

\begin{figure}
\begin{centering}
\includegraphics[width=0.9\columnwidth]{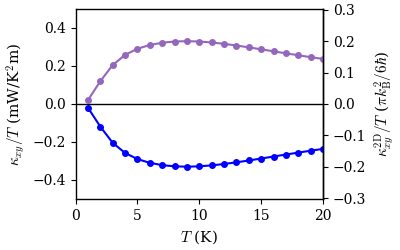}
\par\end{centering}
\caption{Thermal Hall conductivity $\kappa_{xy}/T$ due to magnons in the polarized state as a function of temperature $T$ at field strength $h/S|K|=0.34$, with the interaction parameters $(K,\Gamma,\Gamma')=(-1,0.5,-0.1)$. The magnetic field is applied along the $a$ (-$a$) direction for the blue (purple) curve. The two dimensional thermal Hall conductance $\kappa_{xy}^{2D}/T\equiv\kappa_{xy}d/T$ is also indicated. \label{fig:thc-signs}}
\end{figure}

\section{Discussion}

In the present work, we investigated the influence of topological magnons on thermal Hall transport in quantum magnets. We first presented an analytical proof showing that for an arbitrary 2D commensurate magnetic order as well as the field-polarized state, $\kappa_{xy}$ due to magnons changes sign whenever the magnetic field is reversed, irrespective of the initial orientation of the field. Partly motivated by a recent experiment on $\alpha$-RuCl$_3$\cite{2001.01899}, we investigated this effect numerically in a microscopic model for Kitaev magnets. For example, we demonstrated that the magnon $\kappa_{xy}$ of the 6-site magnetic order and field-polarized state changes the sign when the field orientation is changed from $a$ to $-a$ directions. We emphasize that even though this theorem guarantees the sign change, it does not necessarily lead to a finite magnitude. In the microscopic model for Kitaev magnets, we showed that the finite magnitude of $\kappa_{xy}$ arises due to the topological nature of the magnons, or the Berry curvature of the magnon bands. When the field is applied along $b$-direction, $\kappa_{xy}$ is generally non-zero for commensurate ordered states, in contrast to the vanishing $\kappa_{xy}$ in the case of the field-polarized state reported earlier\cite{chern2020sign} (as also observed in the experiment\cite{2001.01899}). 

In the experiment on $\alpha$-RuCl$_3$, $\kappa_{xy}$ was also measured under tilted magnetic fields and it was found that the sign change of $\kappa_{xy}$ occurs upon reversing only the in-plane field component, while preserving the out-of-plane field component \cite{2001.01899}. In this case, our theorem in Section II does not directly apply. On the other hand, we demonstrated that, as long as the in-plane field contribution to the Zeeman energy dominates, the sign change of $\kappa_{xy}$ still occurs. In addition, the plateau-like behaviour (or very slow change) of $\kappa_{xy}$ in the partially-polarized state as a function of magnetic field strength also arises in this situation. We note, however, that the value of $\kappa_{xy}$ is not quantized, and it varies as a function of temperature. 

One notable difference between our results and the published data on $\alpha$-RuCl$_3$ lies in the ZZ ordered region of the tilted magnetic field experiments. The sign reversal of $\kappa_{xy}$ upon flipping of the magnetic field (our theorem in Section II) applies for any commensurate order. Then, for a 2D ZZ order, there must also be a sign change in $\kappa_{xy}$, assuming that the in-plane contribution to the Zeeman energy dominates. In the experimental data, however, this sign change was absent in the ZZ region. A possible explanation of this difference lies in the 3D nature of the ZZ order. A neutron scattering (NS) experiment\cite{balz2020fieldinduced} confirmed that below the critical field, the ZZ ordering was 3D, while the polarized regime was 2D with a small interlayer correlation length. These NS measurements imply that our results are applicable for the polarized region, whereas our calculations for $\kappa_{xy}$ in the 2D ZZ ordered state may not directly apply. In order to resolve this issue, one may have to take into account possibly complex interlayer couplings to model the 3D ZZ order. The importance of the 3D nature of the ZZ order thus remains to be investigated in future studies. 

While there is always a relative sign between the magnon thermal Hall conductivities under opposite field directions according to Theorem 1, the overall sign of $\kappa_{xy}$ in each of these field directions really depends on the choice of model parameters. We have shown in Sec.~V an example parameter set that yields overall signs of $\kappa_{xy}$ consistent with those observed in Yokoi \textit{et al}\cite{2001.01899,signconvention}. Moreover, under in-plane fields along the $\pm a$ directions, our calculated $\kappa_{xy}$ due to magnons in the polarized state shows a monotonic increase in magnitude as a function of temperature in the low temperature regime, and vanishes in the zero temperature limit. These features agree with the recent experimental data in Ref.~\onlinecite{czajka2021oscillations}. Thus, the presence of a half-quantized thermal Hall conductivity at very low temperatures remains as the ultimate test for the non-Abelian Kitaev spin liquid.

\begin{acknowledgments}
We thank Takasada Shibauchi for helpful discussions. We acknowledge support from the Natural Sciences and Engineering Research Council of Canada (NSERC). E.Z.Z. was further supported by the NSERC Canada Graduate Scholarships - Doctoral (CGS-D), and L.E.C by the Ontario Graduate Scholarship. Y.B.K. was further supported by the Killam Research Fellowship from the Canada Council for the Arts and the Center for Quantum Materials at the University of Toronto. Most of the computations were performed on the Cedar and Niagara clusters, which are hosted by WestGrid and SciNet in partnership with Compute Canada.
\end{acknowledgments}


\begin{thebibliography}{67}
\bibitem [{sig()}]{signconvention}%
  \BibitemOpen
  \href@noop {} {}\bibinfo {note} {In our earlier work, Ref.~[43], we used
  model parameters such that the overall sign of $\kappa_{xy}$ would be
  consistent with the sign convention used in Ref.~[25]. In Sec.~V of this
  work, the overall sign of $\kappa_{xy}$ resulting from the parametrization
  $(K,\Gamma,\Gamma')=(-1,0.5,-0.1)$ would be consistent with the updated sign
  convention in Ref.~[39].}\BibitemShut {Stop}%
\end{thebibliography}
\end{document}


\title{Supplemental Materials}

\author{Emily Zhang}
\affiliation{Department of Physics, University of Toronto, Toronto, Ontario M5S 1A7, Canada}

\author{Li Ern Chern}
\affiliation{Department of Physics, University of Toronto, Toronto, Ontario M5S 1A7, Canada}

\author{Yong Baek Kim}
\affiliation{Department of Physics, University of Toronto, Toronto, Ontario M5S 1A7, Canada}

\maketitle

\setcounter{equation}{0}
\setcounter{figure}{0}
\setcounter{table}{0}

\renewcommand{\thesection}{S\arabic{section}}
\renewcommand{\theequation}{S\arabic{equation}}
\renewcommand{\thefigure}{S\arabic{figure}}
\renewcommand{\thetable}{S\arabic{table}}

\onecolumngrid

\section{Examples of the Magnetic Orders}

We provide the spin configurations for each phase observed on the phase diagram in the main text in Figs. \ref{fig:zza}-\ref{fig:10c}. Each spin configuration was obtained using classical simulated annealing. The examples shown in each figure used external fields under an in-plane field $\mathbf{h}=h\cos\theta\left[11\bar{2}\right]+h\sin\theta\left[\bar{1}10\right]$. 

\begin{figure*}
\begin{centering}
\includegraphics[width=0.6\columnwidth]{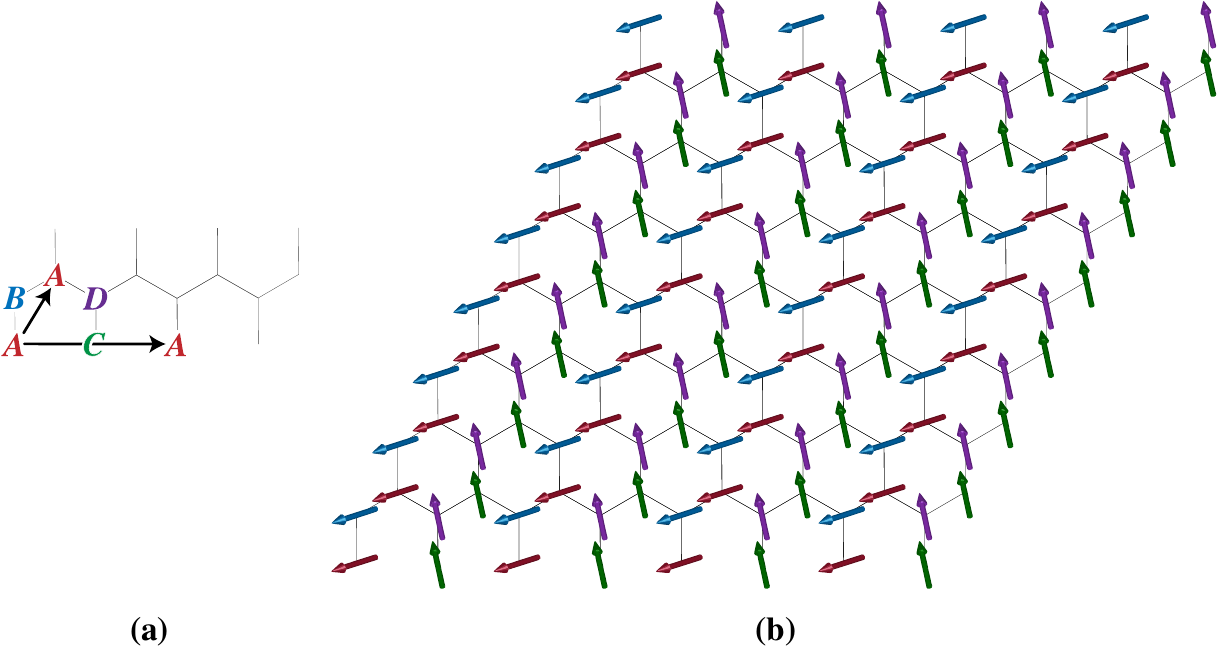}
\par\end{centering}
\caption{(a) The sublattice structure for the ZZa unit cell with the lattice
vectors shown in black. (b) The spin configuration of the ZZa order
shown on the honeycomb lattice. The parameters used were $K=-1,$
$\Gamma=0.5$, $\Gamma'=-0.02$, $h/S|K|=0.15$, and $\theta=150^{\circ}$.
\label{fig:zza}}

\end{figure*}

\begin{figure*}
\begin{centering}
\includegraphics[width=0.6\columnwidth]{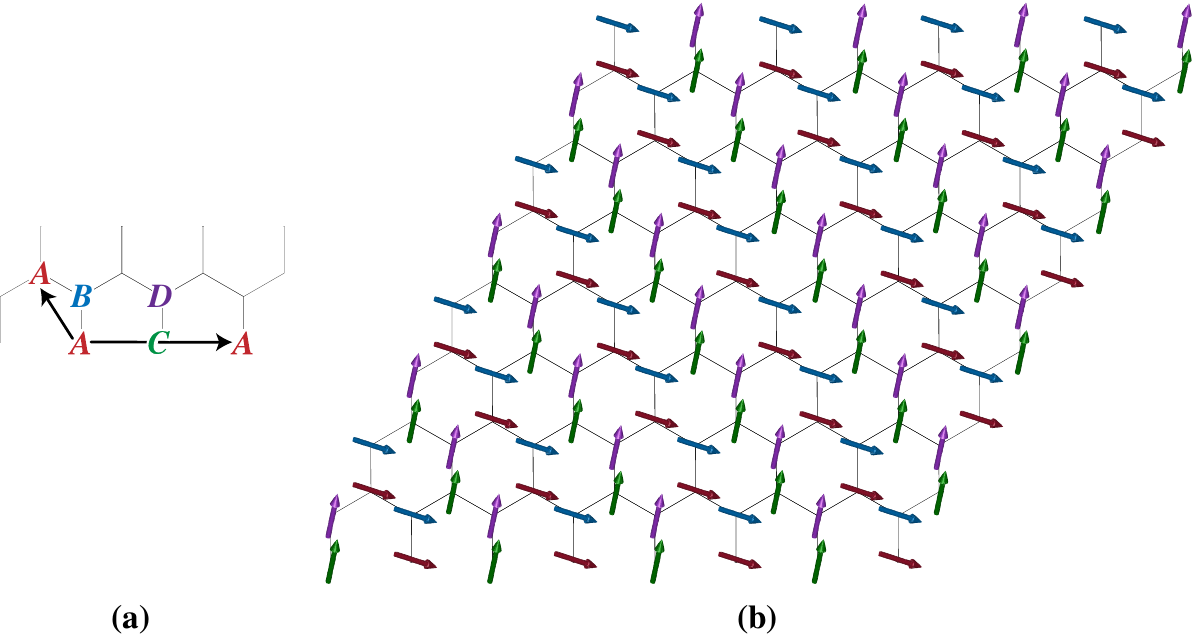}
\par\end{centering}
\caption{(a) The sublattice structure for the ZZb unit cell with the lattice
vectors shown in black. (b) The spin configuration of the ZZb order
shown on the honeycomb lattice. The parameters used were $K=-1,$
$\Gamma=0.5$, $\Gamma'=-0.02$, $h/S|K|=0.15$, and $\theta=30^{\circ}$.
\label{fig:zzb}}
\end{figure*}

\begin{figure*}
\begin{centering}
\includegraphics[width=0.6\columnwidth]{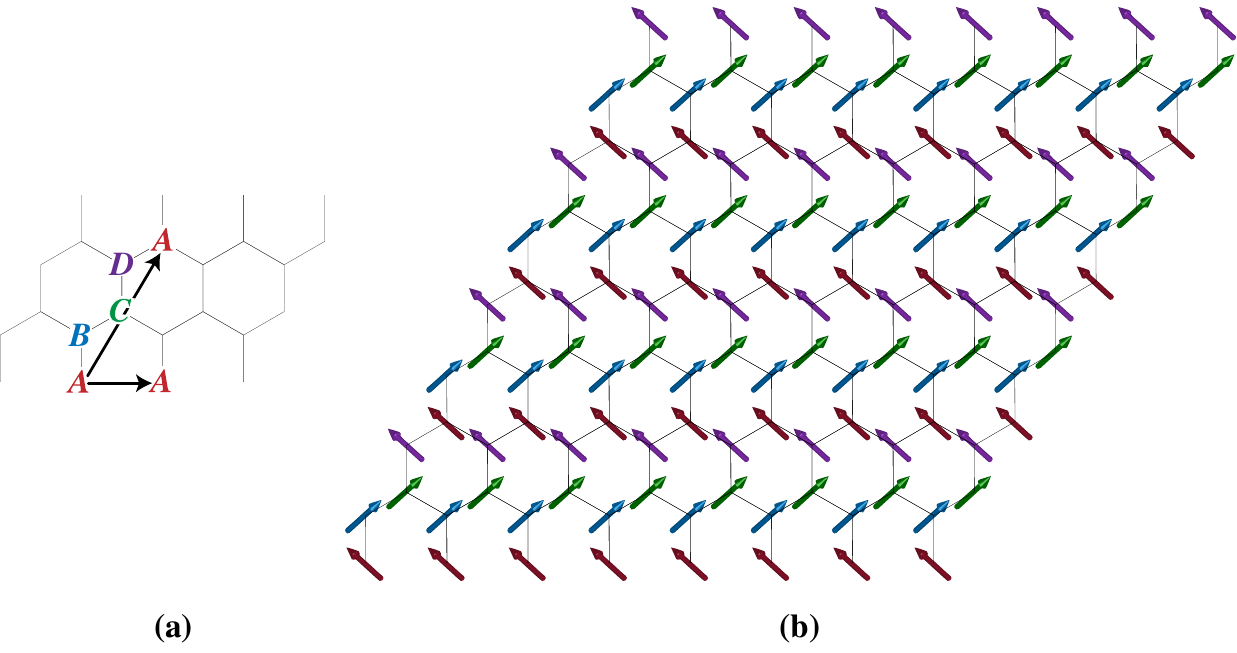}
\par\end{centering}
\caption{(a) The sublattice structure for the ZZc unit cell with the lattice
vectors shown in black. (b) The spin configuration of the ZZc order
shown on the honeycomb lattice. The parameters used were $K=-1,$
$\Gamma=0.5$, $\Gamma'=-0.02$, $h/S|K|=0.15$, and $\theta=90^{\circ}$.
\label{fig:zzc}}
\end{figure*}

\begin{figure*}
\begin{centering}
\includegraphics[width=0.6\columnwidth]{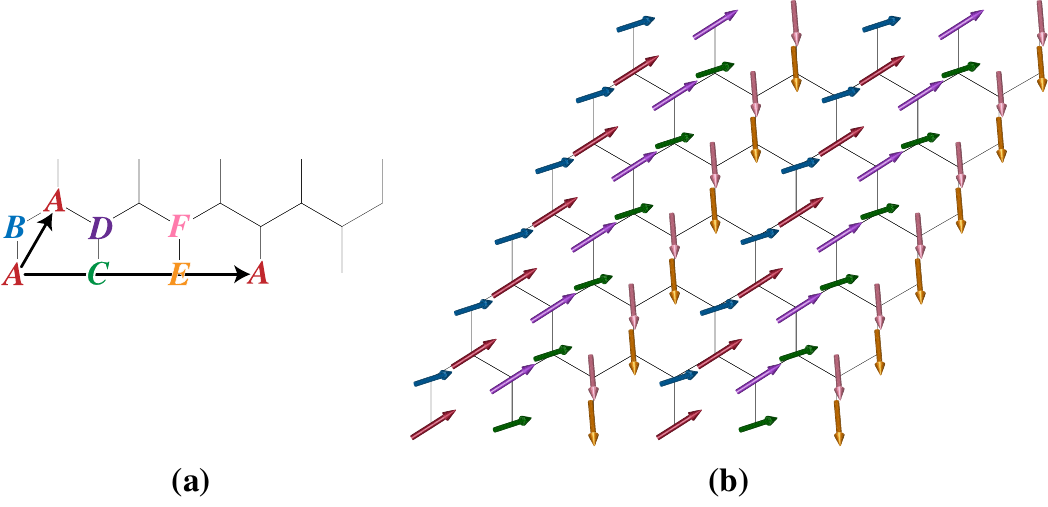}
\par\end{centering}
\caption{(a) The sublattice structure for the 6a unit cell with the lattice
vectors shown in black. (b) The spin configuration of the 6a order
shown on the honeycomb lattice. The parameters used were $K=-1,$
$\Gamma=0.5$, $\Gamma'=-0.02$, $h/S|K|=0.15$, and $\theta=0^{\circ}$.
\label{fig:6a}}
\end{figure*}

\begin{figure*}
\begin{centering}
\includegraphics[width=0.6\columnwidth]{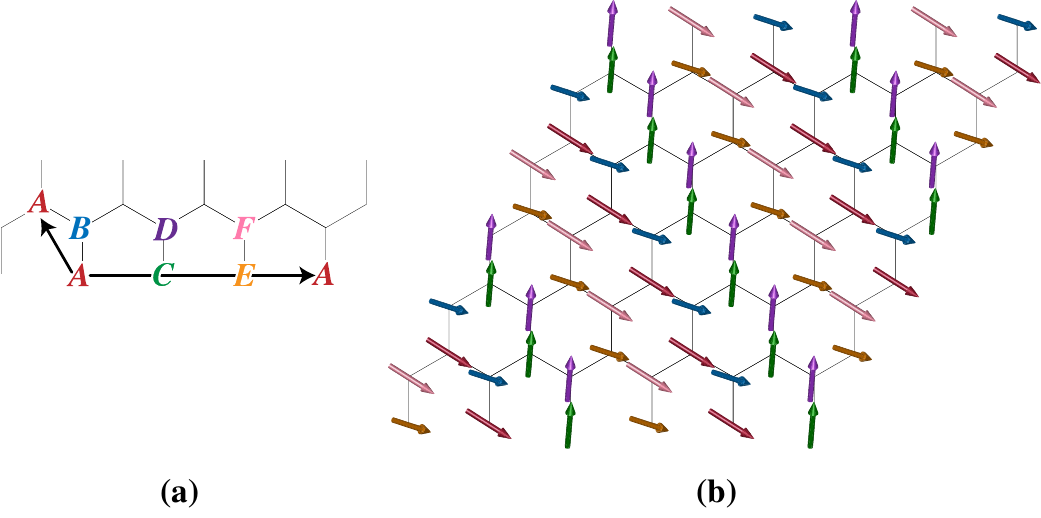}
\par\end{centering}
\caption{(a) The sublattice structure for the 6a unit cell with the lattice
vectors shown in black. (b) The spin configuration of the 6a order
shown on the honeycomb lattice. The parameters used are the same as
the 6a configuration, since the 6a and 6b are degenerate at $\theta=0^{\circ}$.
\label{fig:6b}}
\end{figure*}

\begin{figure*}
\begin{centering}
\includegraphics[width=0.6\columnwidth]{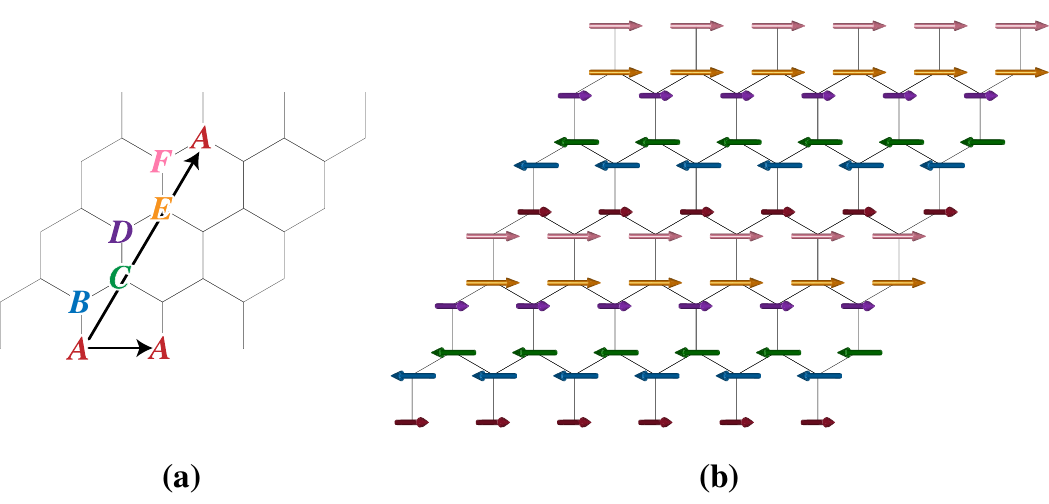}
\par\end{centering}
\caption{(a) The sublattice structure for the 6c unit cell with the lattice
vectors shown in black. (b) The spin configuration of the 6c order
shown on the honeycomb lattice. The parameters used were $K=-1,$
$\Gamma=0.5$, $\Gamma'=-0.02$, $h/S|K|=0.05$, and $\theta=0^{\circ}$.
\label{fig:6c}}
\end{figure*}

\begin{figure*}
\begin{centering}
\includegraphics[width=0.6\columnwidth]{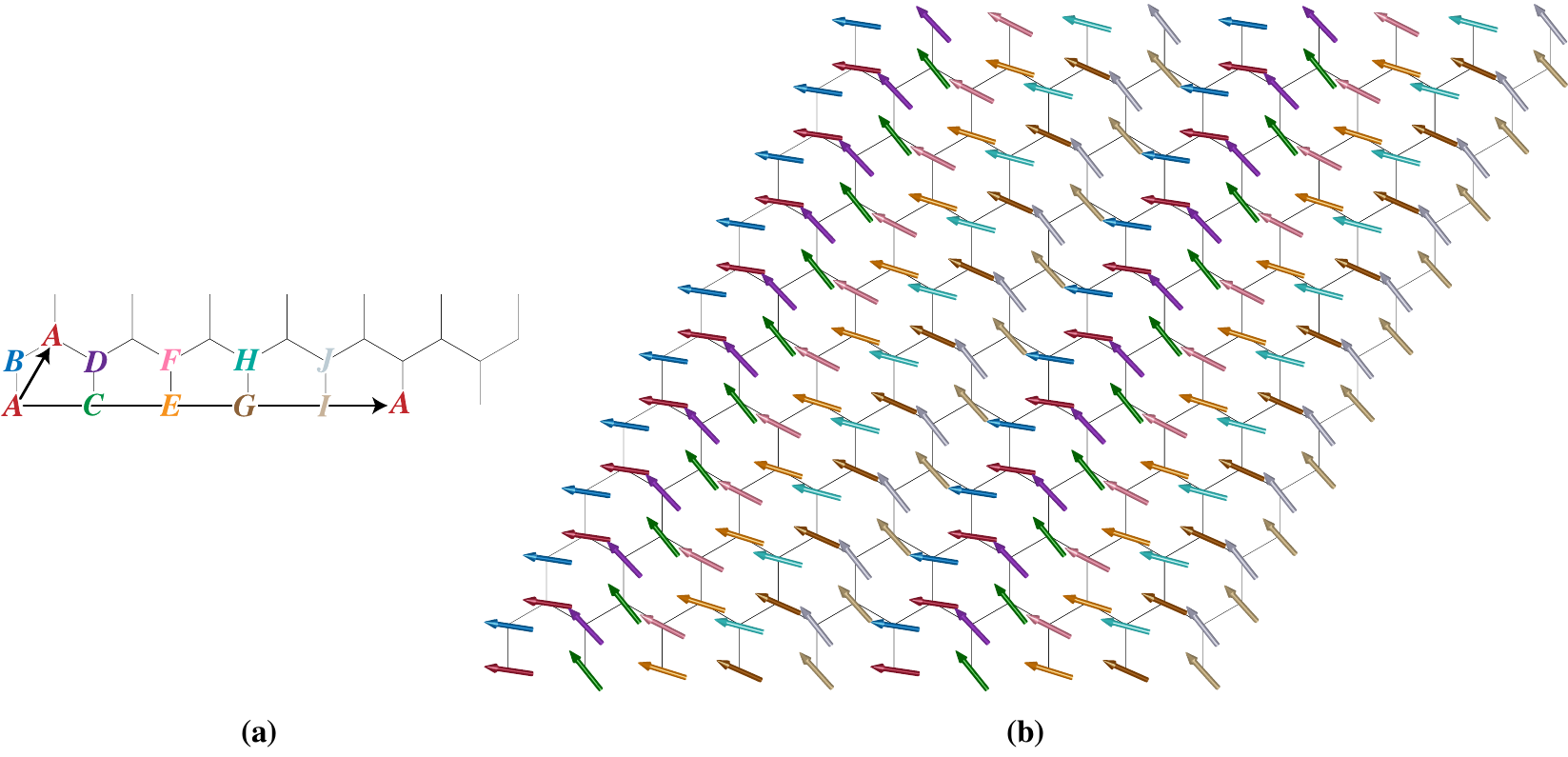}
\par\end{centering}
\caption{(a) The sublattice structure for the 10a unit cell with the lattice
vectors shown in black. (b) The spin configuration of the 10a order
shown on the honeycomb lattice. The parameters used were $K=-1,$
$\Gamma=0.5$, $\Gamma'=-0.02$, $h/S|K|=0.25$, and $\theta=150^{\circ}$.
\label{fig:10a}}
\end{figure*}

\begin{figure*}
\begin{centering}
\includegraphics[width=0.6\columnwidth]{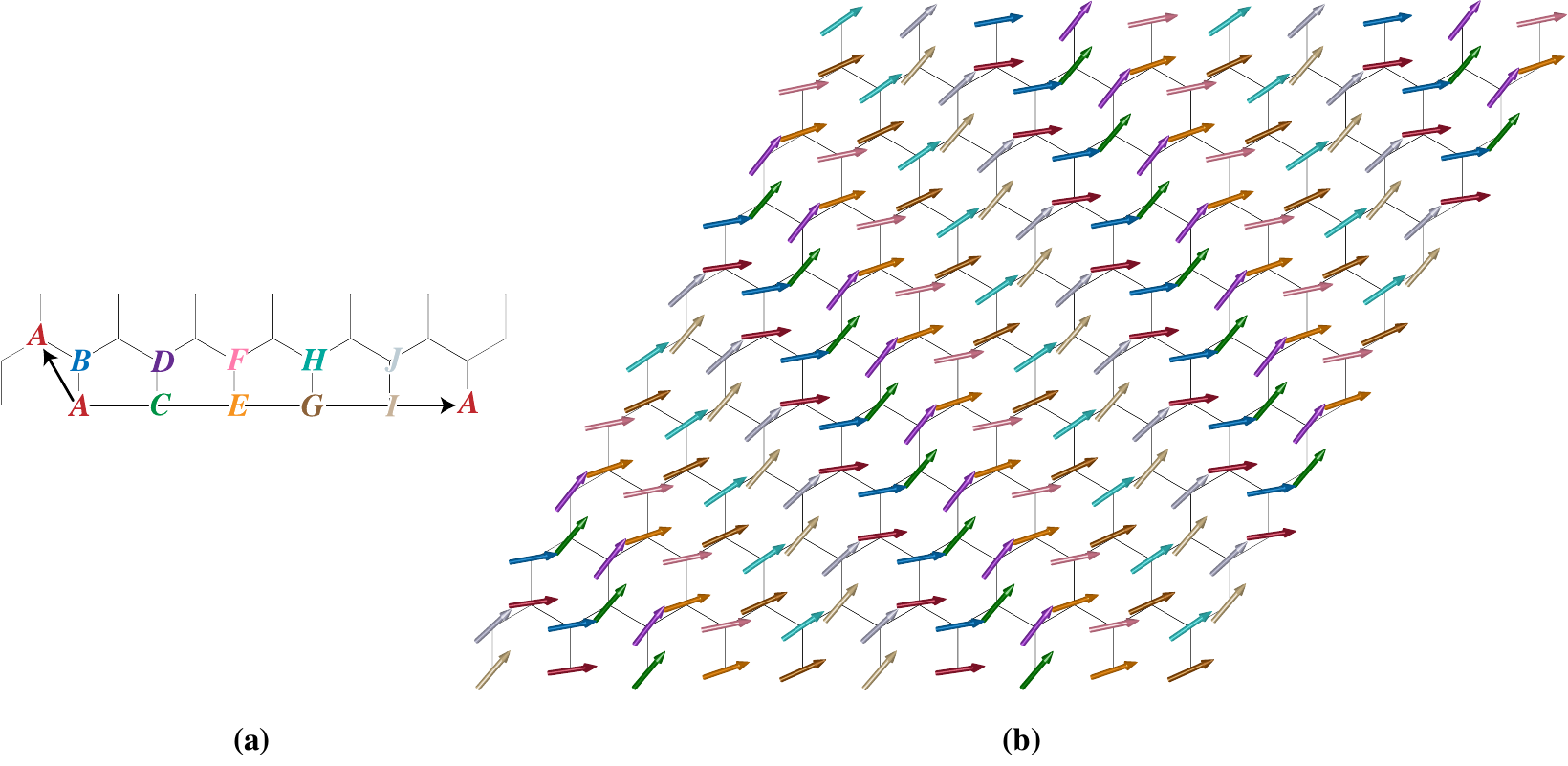}
\par\end{centering}
\caption{(a) The sublattice structure for the 10b unit cell with the lattice
vectors shown in black. (b) The spin configuration of the 10b order
shown on the honeycomb lattice. The parameters used were $K=-1,$
$\Gamma=0.5$, $\Gamma'=-0.02$, $h/S|K|=0.25$, and $\theta=30^{\circ}$.
\label{fig:10b}}
\end{figure*}

\begin{figure*}
\begin{centering}
\includegraphics[width=0.6\columnwidth]{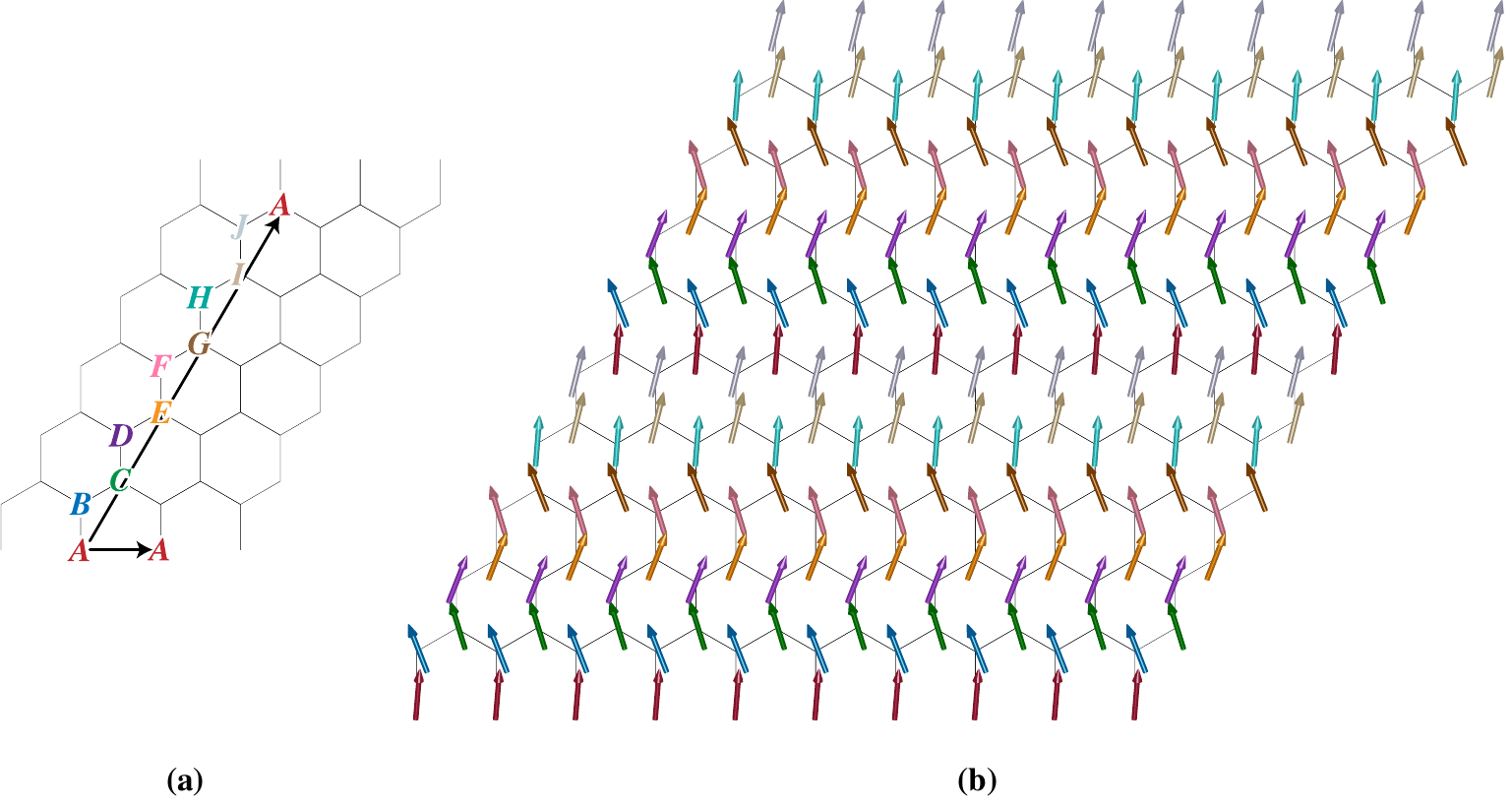}
\par\end{centering}
\caption{(a) The sublattice structure for the 10c unit cell with the lattice
vectors shown in black. (b) The spin configuration of the 10c order
shown on the honeycomb lattice. The parameters used were $K=-1,$
$\Gamma=0.5$, $\Gamma'=-0.02$, $h/S|K|=0.25$, and $\theta=90^{\circ}$.
\label{fig:10c}}
\end{figure*}